\documentclass[runningheads]{llncs}
\usepackage[T1]{fontenc}
\usepackage{cite}
\usepackage{amsmath,amssymb,amsfonts}
\usepackage{algorithm}
\usepackage{algorithmic}
\usepackage{graphicx}
\usepackage{textcomp}
\usepackage{url}
\usepackage{xcolor}
\usepackage{tikz}
\usepackage[normalem]{ulem}

\usetikzlibrary{shapes, arrows.meta, positioning, calc}

\def\BibTeX{{\rm B\kern-.05em{\sc i\kern-.025em b}\kern-.08em
    T\kern-.1667em\lower.7ex\hbox{E}\kern-.125emX}}
    
\usetikzlibrary{positioning}

\usepackage{subfig}
\usepackage[columnwise,running,switch]{lineno} 
\newif\ifannotated

\usepackage{resources/mathcommands}

\usepackage{resources/scheduling}

\begin{document}

\title{Fixed-Priority and EDF Schedules for ROS~2 Graphs on Uniprocessor}

\author{%
Oren Bell\inst{1} \and
Harun Teper\inst{2} \and
Mario G\"unzel\inst{2} \and
Chris Gill\inst{1} \and
Jian-Jia Chen\inst{2}
}%
\institute{%
Washington University in St Louis, USA,\\
\email{\{oren.bell, cdgill\}@wustl.edu}
\and
TU Dortmund University, Germany,\\
\email{\{harun.teper, mario.guenzel, jian-jia.chen\}@tu-dortmund.de}
}


\maketitle

\begin{abstract}
This paper addresses limitations of current scheduling methods in the Robot Operating System (ROS)~2, focusing on scheduling tasks beyond simple chains and analyzing arbitrary Directed Acyclic Graphs (DAGs). While previous research has focused mostly on chain-based scheduling with ad-hoc response time analyses, we propose a novel approach using the events executor to implement fixed-job-level-priority schedulers for arbitrary ROS~2 graphs on uniprocessor systems.

We demonstrate that ROS~2 applications can be abstracted as forests of trees, enabling the mapping of ROS~2 applications to traditional real-time DAG task models. Our usage of the events executor requires a special implementation of the events queue and a communication middleware that supports LIFO-ordered message delivery, features not yet standard in ROS~2. We show that our implementation generates the same schedules as a conventional fixed-priority DAG task scheduler, in spite of lacking access to the precedence information that usually is required. This further closes the gap between established real-time systems theory and ROS~2 scheduling analyses.
\end{abstract}


\section{Introduction}

The Robot Operating System (ROS)~2 is a widely adopted software framework for developing robotics applications, which has garnered significant attention from the real-time systems research community. While ROS~2 provides mechanisms for modular development using nodes and callbacks, its default scheduling capabilities present challenges for predictable real-time performance, particularly for complex applications involving branching or merging data flows (i.e., arbitrary graphs).

Since 2019, there has been ongoing research into modeling and analyzing the end-to-end timing behavior of ROS~2 applications. The default ROS~2 scheduler, called an \emph{executor}, has inherent limitations~\cite{casini2019response,tang2020response,blass2021ros,teper2022end,teper2023timing} that cause highly pessimistic response times for tasks within an application, compared to classical real-time periodic schedulers. There has been work to provide response time analysis~\cite{casini2019response,tang2020response,teper2022end} and strategies to minimize response times~\cite{tang2020response,blass2021ros,choi2021picas}.
The limitations of the default executor are discussed further in Section~\ref{sec:background}.


An alternate executor design, the events executor~\cite{eventexecutor}, was mainlined into ROS~2 in 2023. It is detailed further in Section~\ref{sec:events_executor}.
Prior work~\cite{teper2025reconciling} proposed applying classical real-time theory to the events executor. They showed that the events executor could be used to implement arbitrary periodic fixed-priority scheduling. They also provided some preliminary work providing a response-time analysis for basic chain applications.
Building upon this, our work addresses predictable scheduling for arbitrary ROS~2 graphs. We demonstrate mapping ROS~2 applications to standard DAG task models and propose a novel events queue implementation for the Events Executor. This enables any Limited Preemption Fixed Job-Level Priority (LP-FJP) scheduler -- e.g., Rate Monotonic (RM) or Earliest Deadline First (EDF).

\textbf{Contributions.}
In this paper, we
(i) provide a mapping showing that an arbitrary ROS~2 graph can be modelled as a set of DAG tasks for analysis purposes (Section~\ref{sec:system_model}), including showing that ROS~2 graphs are logically trees for the purposes of scheduling (Section~\ref{sec:graph_unfolding});
(ii) introduce an events queue implementation for the Events Executor that can realize any LP-FJP scheduler for the resulting DAG tasks (Section~\ref{sec:executor_design}), noting that our approach relies on LIFO-ordered message queues which are not yet standard in ROS~2 middleware; and
(iii) provide a formal analysis proving the correctness of this events queue mechanism (Section~\ref{sec:analysis}).
This work further closes the gap between established real-time systems theory and practical ROS~2 scheduling analysis, moving beyond the ad-hoc approaches and chain-structured applications of existing literature.

\section{Background}\label{sec:background}

The Robot Operating System (ROS)~2 provides a component-based framework for building robotic systems. A sample application is illustrated in Figure~\ref{fig:sample_application}. Key concepts include:
\textbf{Nodes} that encapsulate related functionalities (e.g., sensor driver, path planner); 
\textbf{Topics}, which are named communication channels facilitating data exchange via a publish-subscribe paradigm provided by the Data Distribution Service (DDS);
\textbf{Callbacks} are functions within nodes triggered by specific events, primarily 
\textbf{Timer Callbacks} that execute periodically based on a configured rate and 
\textbf{Subscription Callbacks} that execute upon receiving a message on a subscribed topic.
The sequence of callbacks triggered by data flowing through topics forms the application's processing logic, often creating implicit dependencies or chains of execution.

\begin{figure}
    \centering
    \includegraphics[width=\linewidth]{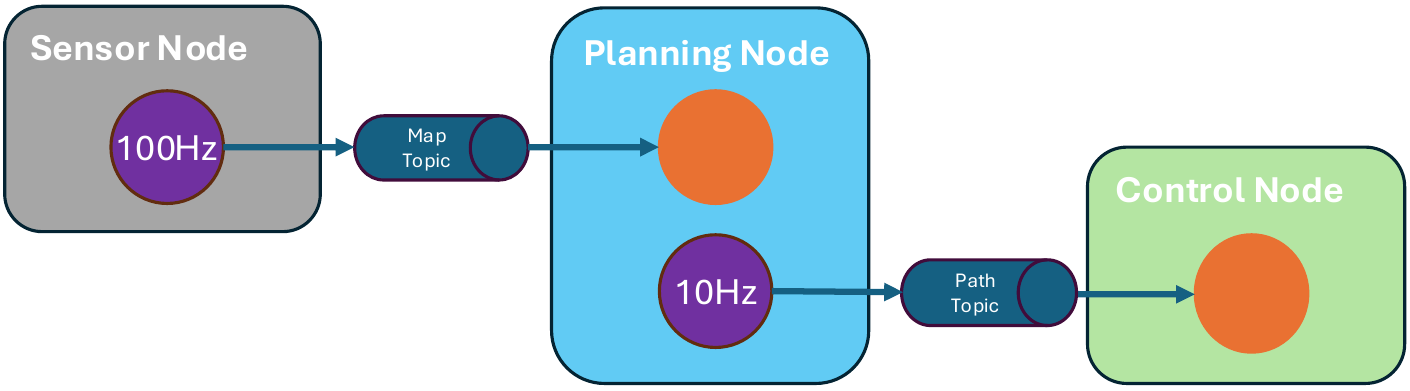}
    \caption{Sample ROS2 Application. Subscription Callbacks illustrated in orange, Timer Callbacks illustrated in purple. Topics are pipes between nodes.}
    \label{fig:sample_application}
\end{figure}


\subsection{Default ROS~2 Executor}
ROS~2 uses an entity called an \emph{executor} to manage the execution of callbacks. The default executor, often referred to as the `SingleThreadedExecutor', has several characteristics and limitations documented in prior work~\cite{casini2019response,tang2020response,blass2021ros,teper2022end,teper2023timing}.
The default ROS~2 executor uses a mechanism called a \textbf{Ready Set} to hold callbacks ready for execution. The ready set is not aware of individual instances of callbacks' releases. It only indicates if a callback is ready for execution, and not how many instances are in the backlog.
The ready set is populated only at specific \textbf{Polling Points}. Once populated, the executor processes callbacks from the set until it is empty before polling again.
The interval between two polling points is known as the \textbf{Processing Window}. Callbacks becoming ready \emph{during} this window must wait until the next polling point, potentially leading to significant delays and priority inversions, as lower-priority callbacks already in the ready set execute first.
The default ROS~2 executor employs a round-robin \textbf{Scheduling Policy} within the ready set, with a static bias favoring timer callbacks over subscription callbacks, and only executing each callback once per processing window, regardless of the number of instances released.
These limitations make it difficult to provide strong real-time guarantees, especially for complex applications.

All these concepts are illustrated in Figure~\ref{fig:example_ros2_default}. At time 0, three timer callbacks are released, and enqueued into the ready set. They are executed in round robin order until time 23, when the ready set is empty. This causes another polling point to occur, and another instance of callback $\tau_1$ is enqueued and executed. This is despite the fact that there were two instances of $\tau_1$ that should have executed. The second instance instance, release at time 20, is not acknowledged and implicitly dropped.


\subsection{The Events Executor}
\label{sec:events_executor}


To address the shortcomings of the default executor, the Events Executor was introduced~\cite{eventexecutor}. Key differences are as follows.
The Events Executor replaces the ready set with an \textbf{Events Queue}. Each triggering event (timer expiration, message arrival) generates a distinct entry in the queue, allowing multiple instances of the same callback to be queued independently if necessary.
The Events Executor uses \textbf{Continuous Polling} to check for new events (polls) after \emph{each} callback execution finishes, rather than waiting for the queue to become empty. This eliminates the problematic processing window of the default executor.
By default, the events queue operates in a \textbf{First-In, First-Out (FIFO)} manner.

This behavior is illustrated in Figure~\ref{fig:example_events}. Unlike with the default executor in Figure~\ref{fig:example_ros2_default}, the additional instance of task $\tau_2$ available at time 20 is acknowledged and enqueued into the events queue. No jobs are dropped because of the processing window.

Teper et al.~\cite{teper2025reconciling} showed that the structure of the events executor allows for easier application of standard real-time analysis (demonstrated for chains) and demonstrated that the events queue itself could be replaced to implement different scheduling policies. Our work builds directly on this potential.

\begin{figure*}
  \centering
  \subfloat[Under ROS~2 default executor. Polling points in dotted red line.\label{fig:example_ros2_default}]{
    \begin{tikzpicture}[yscale=0.4, xscale=0.15]
      \grid{0}{5}{30}{0.2}{5}
      \begin{scope}[shift={(0,4)}] 
        \taskname{$\tau_1$}
        
        \timeline{0}{31}{}
        
        \releases{0,10,30}
        \draw[->, \releasearrowprops, dotted, blue] (20,0) -- (20, \releasearrowlength);
        \exec{0}{3}
        \exec{23}{26}
      \end{scope}
  
      \begin{scope}[shift={(0,2)}] 
        \taskname{$\tau_2$}
        
        \timeline{0}{31}{}
        
        \releases{0,30}
        
        \exec{3}{13}
        
      \end{scope}
      
      \begin{scope}[shift={(0,0)}] 
        \taskname{$\tau_3$}
        
        \timeline{0}{31}{}
        \labelling{0}{30}{5}{0}
        
        \releases{0,30}
        
        \exec{13}{23}
      \end{scope}

      \draw[dotted, red, very thick] (0,-.7) -- (0,5.7);
      \draw[dotted, red, very thick] (23,-.7) -- (23,5.7);

      \end{tikzpicture}
  }\hfill
  \subfloat[Under events executor.\label{fig:example_events}]{
    \begin{tikzpicture}[yscale=0.4, xscale=0.15]
      \grid{0}{5}{30}{0.2}{5}
      \begin{scope}[shift={(0,4)}] 
        \taskname{$\tau_1$}
        
        \timeline{0}{31}{}
        \releases{0,10,20,30}
        \foreach \x in {20}{
          \draw[draw=blue, ->, \releasearrowprops] (\x,0) -- (\x, \releasearrowlength);
        }
        \exec{0}{3}
        \exec{23}{26}
        \exec[fill=blue]{26}{29}
      \end{scope}
  
      \begin{scope}[shift={(0,2)}] 
        \taskname{$\tau_2$}
        
        \timeline{0}{31}{}
        
        \releases{0,30}
        
        \exec{3}{13}
        
      \end{scope}
      
      \begin{scope}[shift={(0,0)}] 
        \taskname{$\tau_3$}
        
        \timeline{0}{31}{}
        \labelling{0}{30}{5}{0}
        
        \releases{0,30}
        
        \exec{13}{23}
      \end{scope}


      \end{tikzpicture}
  }
  \caption{Example schedules contrasting default ROS~2 executor with events executor}
\end{figure*}
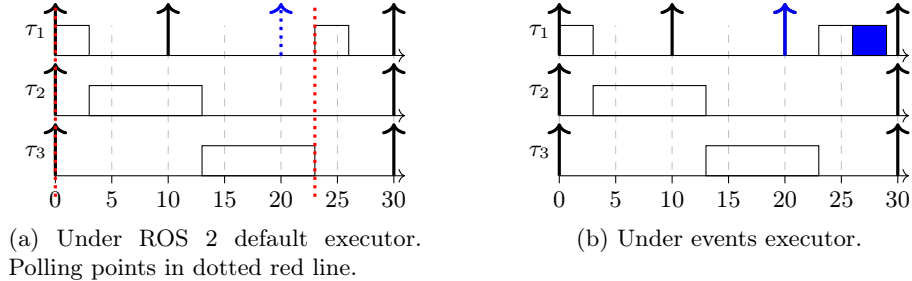

\section{Related Work}\label{sec:related_work}

\textbf{Executor Extensions and Replacements.} Recognizing the default executor's limitations, researchers proposed modifications and alternatives. PiCAS~\cite{choi2021picas} introduced a priority-driven chain-aware a priori work assignment algorithm, often requiring decomposition of the application across multiple executor instances. Other works explored different executor designs~\cite{arafat2022response}. The Events Executor~\cite{eventexecutor}, as discussed in Section~\ref{sec:background}, represents a significant step towards a more standard scheduling infrastructure within ROS~2.

\textbf{Bridging the Gap for General Graphs (Our Contribution).} While~\cite{teper2025reconciling} laid important groundwork and demonstrated the potential for chains, it did not provide a concrete mechanism or analysis for applying standard priority-based scheduling (like LP-FJP) to general ROS~2 \emph{graphs} (involving branching and merging). This paper fills that gap by: (1) presenting the graph unfolding technique to map ROS~2 graphs to standard DAG models for analysis, (2) designing a specific two-queue mechanism for the Events Executor's queue that equates to LP-FJP scheduling, and (3) providing a formal proof that this mechanism correctly schedules tasks according to their job-level priorities \emph{without} needing explicit runtime access to the graph structure, relying instead on implicit priority inheritance. This allows direct application of existing LP-FJP DAG scheduling analysis, such as~\cite{nasri2019response}.


\section{System Model}\label{sec:system_model}

In this section, we discuss the unique manner in which ROS~2 schedules callbacks, and the assumptions we must make before mapping applications to a conventional DAG task model, detailed next.

\subsection{Task Model}\label{sec:task_model}

Given a set $\mathbb{T} = \{\tau_1, \tau_2, ..., \tau_n\}$ of tasks, a \textbf{periodic DAG task} $\tau_i$ is specified by the tuple $\tau_i = (T_i, D_i, \phi_i) \in \mathbb{R}^3$, where $T_i>0$ is the period, $D_i > 0$ is the relative deadline, and $\phi_i$ is the phase. The periodic task $\tau_i$ releases its first job (task instance) at time~$\phi_i$, and subsequent jobs are released every $T_i$ time units.  Every job has an absolute deadline specified as its release time plus the relative deadline.

A \textbf{sporadic DAG task} $\tau_i$ is similarly defined as the tuple $\tau_i = (T_i, D_i) \in \mathbb{R}^2$, where $T_i$ is instead the minimum interarrival time between two successive jobs. The \emph{periodic task model} is a specialization of the \emph{sporadic task model}.

Each DAG task is associated with a graph $G_i = (V_i, E_i)$. The nodes in the finite set $V_i = \{v_{i1}, v_{i2}, ...\}$ model non-preemptible portions of work, or \emph{subtasks}, in the task $\tau_i$. Each subtask $v_{ij}$ has a worst-case execution time (WCET) $w_{ij}$ and may have precedence constraints modeled by the relations in the finite set $E_i = \{e_{i1}, e_{i2}, ...\}$. A subtask is released when all of its precedence constraints (if any exist) are met. We define a precedence relation as $v_{ij} \prec v_{ik}$ and say that $v_{ij}$ is a \textbf{parent} of $v_{ik}$.

Each subtask $v_{ij}$ has an associated amount of work, $W_{ij}$, which is the WCET the subtask needs to complete. We say a \textbf{job} of task $\tau_i$ is complete when one instance each of $v_{ij} \in V_i$ are complete. 
%
We focus specifically on DAG task schedulers with \emph{limited preemption} and \emph{fixed job-level priority}.
\emph{Limited preemption} (LP) means that individual subtasks cannot be preempted, but a task can be preempted in between constituent subtasks. We take this approach to conform to the natural behavior of ROS~2 executors, which will not preempt running executables. Thus, the only possible preemption points are at the completion of subtasks.
\emph{Fixed job-level priority} (FJP) means that a task and all its constituent subtasks have the same priority for a given job instance. The priority between two job releases of the same task may vary.
All fixed-priority schedulers are naturally also fixed for individual jobs. Earliest Deadline First (EDF) is an example of a dynamic-priority scheduler whose priorities are fixed for individual jobs, i.e., it is FJP.

\subsection{Mapping ROS~2 to DAG Task Model}
As in prior literature\cite{casini2019response,tang2020response,choi2021picas}, we abstract the dataflow paths of a ROS~2 application into a graph structure. To illustrate, consider an example ROS~2 application with a couple sensor nodes that take periodic measurements, and publish their data to topics subscribed by downstream computation nodes, like sensor fusion, mapping, planning, etc. The sensors are time triggered and so associated with timers. All data-driven tasks are associated with subscriptions.

Within the graph structure, each callback (timer, subscription, etc) forms a subtask, which may be released periodically (in the case of timers), or data-driven (in the case of subscriptions). During the course of a its execution, a subtask $\tau_i$ may publish a message to a topic, and any subscribers to that topic, $\tau_j$ will be modeled as children of $\tau_i$.

We assume that ROS~2 applications do not have such loops, enabling the connection to DAGs. Given no restrictions on the system design, it is possible for subscriptions to send messages to themselves or form a loop in which data is propagated from one subscription through other subscriptions and back to itself. But if a subscription is allowed to continually trigger itself, it creates a positive feedback loop of exploding message bandwidth, so this is strongly discouraged by good ROS~2 design principles.
After establishing these precedence constraints, these subtasks can then be arranged into \emph{DAG tasks}. In the next section, we will demonstrate they have a tree structure, and so a singular subtask can be described as the root. These \emph{DAG tasks} are modeled as \emph{periodic tasks} when a timer is at the root, and externally-driven \emph{sporadic tasks} when a subscription is at the root. 

In the existing DAG scheduling literature, the term \emph{node} is synonymous with \emph{subtask}~\cite{serrano2016response,verucchi2023survey,nasri2019response,DBLP:books/sp/Buttazzo24,fonseca2015multi}.
To avoid confusion, we will strictly use the ROS~2 definition of \emph{node}, which is an object that serves as a programmatical abstraction to bundle executable entities under a unified namespace. \emph{Subtask} will be used exclusively to refer to callbacks in the context of scheduling.

We only concern ourselves with dataflow that occurs through the Data Distribution Service (DDS), a publish-subscribe communication middleware used by default in a ROS~2 based system. Any two subtasks that communicate with shared memory (as is the case with some \emph{fusion node}\footnote{A fusion node is a component with a single output that combines data from multiple inputs. Unlike many-to-one topics, the output topic operates at an independent rate, or a rate matching only one of the inputs.} implementations) are not necessarily modeled as being in the same task. A precedence constraint is only imposed when the execution of one subtask triggers the release of another. Any subtasks that exchange information asynchronously, bypassing the pub/sub mechanism, are not modelled by our work.

Code within callbacks can be fairly arbitrary. In order for the analysis of a ROS~2 application to be feasible, we must impose some constraints on its behavior:
(1) children are eligible for release as soon as parents finish, so no inter-subtask latency needs to be modeled; 
(2) no new threads are created;
(3) no blocking on IO reads occur; and
(4) more generally, all execution times of callbacks are bounded.

Conventional real-time schedulers stipulate that child subtasks are eligible for execution when their parents complete, and thus imposing this requirement is needed to comply with the analysis of a LP-FJP scheduler. We restrict the callback's ability to spawn new threads, since the scope of this paper is limited to the uniprocessor case. According to the model we established in \ref{sec:task_model}, subtasks have a finite WCET. Therefore anything capable of an unbounded delay is not allowed.


One of the most unique aspects of modeling ROS~2 applications as a DAG task is that precedence constraints are logically disjunctive: a subtask will execute a job for \emph{each} execution of \emph{each} of its parents. This is not in line with the conventional DAG task model we laid out above, until we apply the abstraction of virtually unfolding the graph into a forest of trees, described next.

\subsection{Graph Unfolding}\label{sec:graph_unfolding}

Next, we show that ROS~2 applications that can be modeled as DAGs can be converted to trees due to the logically disjunctive precedence constraints.
In conventional DAGs, subtasks will wait for all of their parents before being eligible for execution. In ROS~2, however, when multiple parents publish to the same subtask, the child will execute once for each parent event. To model this behavior \emph{for analysis purposes}, we adapt the concept of graph unfolding~\cite{blum2019language} where we virtually duplicate each subtask so that each resulting virtual subtask instance has exactly one parent.



This transformation is illustrated in Figure~\ref{fig:ros2_unfolding}. This creates a tree (or rather a forest, since there can be multiple trees). We posit some properties of this new representation:
(i) each subtask has either 1 or 0 parents;
(ii) each subtask is the root of a unique subtree;
(iii) subtasks with 0 parents will be the root of a tree; and
(iv) root subtasks can be modelled sporadically.

To justify the last point, we note that timer tasks trigger periodically, and it's assumed that parentless subscription tasks are triggered by some external event for which a minimum inter-arrival time is known. This assumption is in line with prior literature\cite{tang2020response}.
As previously mentioned, this transformation is possible for ROS~2 applications that are free of cycles and can be modeled as DAGs.

\begin{figure}
    \centering
    \subfloat[before unfolding]{\label{fig:ros2_unfolding-a}
        \includegraphics[width=0.47\linewidth]{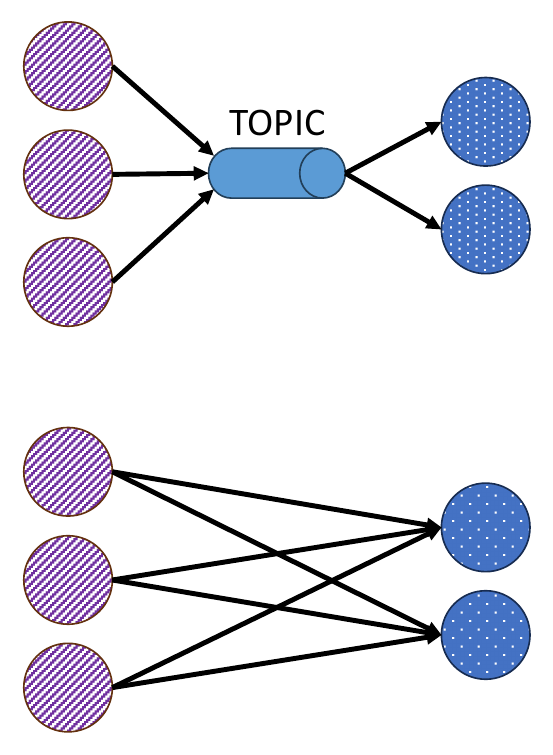}
    }
    \subfloat[after unfolding]{\label{fig:ros2_unfolding-b}
        \includegraphics[width=0.3\linewidth]{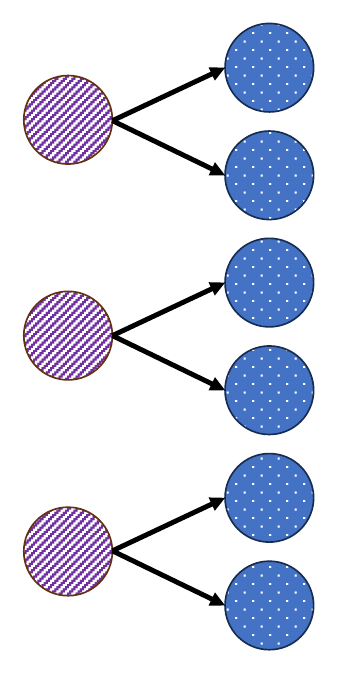}
    }
    \caption{Unfolding of a ROS~2 Graph into a Forest. Timers in purple, subscriptions in blue}
    \label{fig:ros2_unfolding}
\end{figure}


\subsection{Fusion Nodes}

A common pattern in ROS~2 involves \emph{fusion nodes}, which combine data from multiple input topics. While the previous sections already sufficient for general DAGs, we present special commentary for the modeling of fusion nodes in graph unfolding. This depends on their implementation:
\begin{itemize}
    \item \textbf{Trigger-Based Fusion.} If a node subscribes to multiple topics and its single callback executes upon message arrival, potentially varying its behavior or execution time based on which topic triggered it or conditionally publishing, our unfolding process (Section~\ref{sec:graph_unfolding}) still applies. The fusion node is duplicated for each incoming topic edge. While this correctly models the potential execution paths for analysis, applying the node's full WCET to each unfolded instance can lead to pessimistic (safe, but potentially loose) response time bounds if the actual execution is conditional or significantly shorter for certain triggers.
    \item \textbf{Timer-Based Fusion.} If a node subscribes to multiple topics, buffers the data internally, and uses a separate periodic timer callback to process the fused data and publish it, our model handles this naturally. The subscriber callbacks form the ends of their respective upstream DAGs. The timer callback forms the root of a \emph{new} separate DAG, as its execution is triggered by its timer, not directly by the arrival of input messages via DDS.
\end{itemize}

Now that we've established the mapping to a DAG task model, we next describe how to make use of the events executor to create a fixed-job-level-priority scheduler.


\section{Executor Design}\label{sec:executor_design}

Prior work \cite{teper2025reconciling} demonstrated the usefulness of the new ROS~2 events executor for scheduling task chains on a uniprocessor. We now present a queueing mechanism for the events executor that enables ROS~2 to schedule general DAG tasks as well, with fixed-job-level-priority scheduling. 

In \cite{teper2025reconciling}, a scheduler segregates a release thread from a scheduler thread. During DDS events, received messages insert an event object into a priority queue (also introduced in \cite{teper2025reconciling}).
%
The decision about which subtask to release depends entirely on the implementation of the delegator's priority queue. This is one of the primary benefits the events executor has over the original ROS~2 executor. By default, it is a FIFO scheduler. However, by replacing this events queue with compatible alternative implementations, other scheduling paradigms can be created.
%
%
%
%
In a traditional first-task-priority scheduler (FTP), all subtasks in the graph (in our case, chain/tree) inherit the same priority as their parent. In a pub/sub system such as ROS~2, the linkage between parent and child is not readily apparent. Even if a ROS~2 executor tracked the relation between publisher and subscriber, a single topic can have multiple publishers. So a subtask may have multiple parents, with no obvious way to distinguish which triggered the release of an individual job.


We propose a queueing mechanism that allows children to infer which parent triggered their execution and inherit the appropriate priority. It consists of 2 queues:

\begin{itemize}
    \item \textbf{Priority queue for root subtasks.} Released timer subtasks are kept in a min-heap (or max-heap) sorted by some fixed priority, such as their period for a rate-monotonic scheduler. This queue may include externally driven subscribers, if there is a mechanism to specify their priority.
    It is referenced in Algorithms~\ref{alg:release_logic} and \ref{alg:scheduling_logic} as $root\_queue$
    \item \textbf{LIFO queue for child subtasks.} All other subtasks, like subscriptions and services, are held in a LIFO queue. Each entry is paired with a priority, which is initially undefined, but is set on the next scheduling decision. This queue is also sorted by priority, per the discussion in Section~\ref{sec:analysis}.
    It is referenced in Algorithms~\ref{alg:release_logic} and \ref{alg:scheduling_logic} as $child\_queue$
\end{itemize}

The release logic is simply sorting a new subtask into the appropriate queue, as shown in Algorithm \ref{alg:release_logic}. Root subtasks are assumed to have some known priority (based on a period, deadline, or some fixed value), and child subtasks leave this value undefined.

\begin{algorithm}
\caption{Release Logic}
\begin{algorithmic}[1]\label{alg:release_logic}
    \IF{event.priority is defined \emph{(event is timer)}}
        \STATE root\_queue.push(event)
    \ELSIF{event.priority is undefined \emph{(event is subscription)}}
        \STATE event.priority = latest\_priority
        \STATE child\_queue.push(event)
    \ENDIF
\end{algorithmic}
\end{algorithm}

The scheduling logic is detailed in Algorithm \ref{alg:scheduling_logic}. 
Whenever a job is scheduled, its priority is recorded. This is assumed to the be the priority of all subscribers released until the next scheduling decision. Any new subscribers that get released will inherit this priority at the next scheduling decision.

\begin{algorithm}
\caption{Scheduling Logic}
\begin{algorithmic}[1]\label{alg:scheduling_logic}
    
    \IF{root\_queue.top.priority $>$ child\_queue.top.priority} 
        \STATE latest\_priority = root\_queue.top.priority 
        \STATE schedule root\_queue.pop() 
    \ELSE 
        \STATE latest\_priority = child\_queue.top.priority 
        \STATE schedule child\_queue.pop() 
    \ENDIF
\end{algorithmic}
\end{algorithm}

When a subtask is scheduled, the DDS communication layer, which manages messages and communication, is queried for the message data to pass to the subtask's callback.
%
To ensure that the correct instance of a subscriber is released, the underlying DDS layer also should be configured to release messages in LIFO ordering. Subscription events are released in LIFO order, but the messages they operate on are only passed to them after scheduling. Thus, it is crucial that the underlying DDS message queues are also LIFO ordered, since otherwise multiple instances of a single subscription may be assigned the wrong message at scheduling. This would be equivalent to running different jobs of a subtask out of order, which is undesirable.

\section{Analysis}\label{sec:analysis}

In this section, we show that a ROS~2 application running on the events executor with our specialized queue implementation is analytically equivalent to a fixed-priority tree scheduler on a uniprocessor system. That is, the application will generate the same schedule as a conventional fixed-task-priority tree scheduler if given the unfolded equivalent of a ROS~2 application's graph.
%
We focus on limited preemption because ROS~2 executors cannot be preempted during the execution of a subtask, but subtasks from different trees are allowed to interleave, so the task as a whole has limited preemption.

We first assume that subtasks at the root of a tree have some known integer-valued priority for each job release. This may be a fixed priority, or can vary between job releases (such as a deadline for EDF). This priority is inherited by all subtasks.
%
These root subtasks are naturally the first to run. In the following lemmas, we will lay out the natural ordering of scheduled subtasks and the priorities they have in association with each other.

\begin{lemma}\label{lemma:root_subtask_eligibility}
When idle, the only subtasks that may be eligible for release are the root subtasks of each tree.
\end{lemma}

\begin{proof}
Since the system idles, the two queues specified in Algorithms \ref{alg:release_logic} and \ref{alg:scheduling_logic} are empty, by definition of a greedy scheduler. Since we stipulate that all child subtasks are released as soon as their parent complete, no new child tasks will become eligible after the last task has completed. Since the root subtasks are triggered periodically, they are the only ones that can be released when no other subtasks are running.
\end{proof}







\begin{lemma}\label{lemma:child_parent_priority_equivalence}
Any newly released child subtask will have a priority equal to the subtask that just finished execution.
\end{lemma}

\begin{proof}
In our system model we require that children are released immediately upon the parent finishing and no later. Therefore it would be contradictory to assume anything other that the subtask that just finished is the parent of the child subtask that was released. By definition of FJP, a child inherits the priority of its parent.
\end{proof}

\begin{lemma}\label{lemma:released_child_priority}
Any unscheduled subtask, $\tau_a$, in the events queue must have a priority equal or lesser than the currently running subtask, $\tau_b$.
\end{lemma}

$$P_a \leq P_b$$

\begin{proof}
By contradiction:
Otherwise $\tau_a$ would have been scheduled instead. Since our scheduler is nonpreemptive, no new events would have been enqueued by Algorithm \ref{alg:release_logic} since the last scheduling decision.
\end{proof}

\begin{lemma}\label{lemma:lifo_is_priority_insertion}
Any element placed into the children queue from Algorithm \ref{alg:release_logic} will always be the highest priority element.
\end{lemma}

\begin{proof}
Any child subtasks, $\tau_c$, released at the completion of subtask $\tau_b$ must have the same priority as subtask $\tau_b$, by Lemma~\ref{lemma:child_parent_priority_equivalence}.

$$P_b = P_c$$

Therefore, applying Lemma~\ref{lemma:released_child_priority}, any already released subtask $\tau_a$ must have a lower priority.

$$\forall \tau_a, \forall \tau_c \succ \tau_b | P_a \leq P_b$$

$$\implies P_a \leq P_c $$

\end{proof}

We can impose the rule that ties are broken in LIFO order, based on release eligibility.

\begin{lemma}\label{lemma:lifo_is_priority}
Any element removed from the LIFO queue will always be the highest priority element.
\end{lemma}

\begin{proof}
Lemma~\ref{lemma:lifo_is_priority_insertion} combined with definition of LIFO.
\end{proof}

\begin{theorem}\label{theorem:executor_works}
The highest-priority subtask will always be scheduled next by Algorithm \ref{alg:scheduling_logic}.
\end{theorem}

\begin{proof}
In Algorithm \ref{alg:scheduling_logic}, only the top subtasks of the roots queue and children queue are compared. Per Lemma~\ref{lemma:lifo_is_priority}, the children queue is functionally a priority queue. Therefore, the two head subtasks have higher priority than any other subtasks in their respective queues. The greater of the two is guaranteed to be a higher priority than any released subtasks.
Therefore, scheduling Algorithm \ref{alg:scheduling_logic} is guaranteed to schedule the highest priority released job.
\end{proof}

Additionally, this proof considers priorities of individual instances of subtasks to be independent from each other.
Therefore, it's compliant with any job-level fixed-priority scheduler.
This assumes that all constraints we detailed in Section~\ref{sec:system_model} are kept, and the DDS layer has been configured to release messages in LIFO order as well.

\subsection{Extension to non-LIFO queues}
\label{sec:non_lifo_analysis}

The events queue design in Section~\ref{sec:executor_design} above uses LIFO-ordering for new subscriber jobs.
However, we are not aware of any ROS~2 communication middleware vendors that support LIFO message queues. This means that even when jobs are executed in the correct order, the messages for them to process would be given in the wrong order. This would be equivalent to running different jobs of a specific subtask out of order, which is undesirable. This work provides a motivating usecase for the implementation of LIFO message queuing.
Even when LIFO queues are not available, a subset of applications are still schedulable. Specifically, if guarantees can be made that there will never be more than one released job of a task at a time, then the inversion from lacking a LIFO message queue never occurs. Such specific applications are detailed as follows.

\begin{lemma}\label{lemma:one_element_queue}
Any queue that always has at most one element is functionally a LIFO queue.
\end{lemma}

\begin{proof}
In a queue of size one, the only element able to be removed is the head element, so removing it abides by the LIFO definition.
\end{proof}



\begin{lemma}\label{lemma:harmonic_parents}
If a subscription $C$ with response time $R_C$ is the only subscriber for a topic and its parents $A, B, \dots$, with response times $R_A, R_B, \dots$ are harmonic with a \emph{sufficient} phase shift, 
then the message queue for that subscription's topic will never have more than one message in it.
\emph{Sufficient} means that the taskset owning $A, B, C, \dots$ is strictly periodic and 
$$t_{RA} + R_A + R_C \leq t_{RB}$$
$$t_{RB} + R_B + R_C \leq t_{RA} + T_A,$$
where $t_{Rx}$ is the release time of $x$, occuring at timepoints $\phi_x + k T_x \mid k \in \mathbb{N}$. Without loss of generality, this extends to more than 2 parents by applying these conditions to each pair of parents.

\end{lemma}

\begin{proof}
A job/message $x$ is added to the queue at each job release, $t_{Rx}$, and will be removed by the time that job completes $t_{Rx} + R_x$.

The subscription task in question has a release time of $t_{RC} = t_{RA} + R_{A}$, where $A$ is the triggering parent. It is removed from the message queue by the pessimistic timepoint

$$t_{RA} + R_{A} + R_{C}$$

By the given conditions, this is earlier than the release of the next parent, $t_{RB}$. Therefore the queue will be empty by the time the next job is released. And that job will be removed before the next, and so on.

The events queue will only have either zero or one jobs of this subscription.
\end{proof}

\begin{theorem}\label{theorem:executor_works_with_one_element_queue}
Even if the message queue isn't configured to be LIFO order, the highest-priority subtask will always be scheduled next by Algorithm \ref{alg:scheduling_logic} if the conditions formulated in Lemma~\ref{lemma:harmonic_parents} hold.
\end{theorem}

\begin{proof}
If the conditions in Lemma~\ref{lemma:harmonic_parents} are true, then the message queue has a maximum size of one. 
Consequently, due to Lemma~\ref{lemma:one_element_queue}, 
the queue is functionally a LIFO queue. 
Hence, Theorem~\ref{theorem:executor_works} can be applied.
\end{proof}

\subsection{Applying Existing Schedulability Tests and Response Time Analysis}

Our executor design implements a LP-FJP policy, where subtasks run non-preemptively, but preemption can occur between subtasks. Furthermore, as shown in Section~\ref{sec:graph_unfolding}, ROS~2 graphs can be modeled as a forest of DAGs (trees) for analysis. Therefore, established analyses for LP-FJP scheduling of DAG tasks on uniprocessors are directly applicable. 

Specifically, the analysis provided by Nasri et al.~\cite{nasri2019response} addresses this exact model (LP-FJP for DAGs). Their analysis tool~\cite{npscheduabilityanalysis} considers factors like release jitter and execution jitter, and supports multiprocessor analysis, though we focus on the uniprocessor case here. We utilize this specific analysis~\cite{nasri2019response} in Section~\ref{sec:evaluations} to calculate the theoretical worst-case response times (WCRTs) for comparison with our experimental results. Using an analysis designed for the target scheduling policy (LP-FJP DAGs) is crucial for accurate validation, as opposed to analyses developed for different policies like the default ROS~2 executor~\cite{casini2019response}.

When applying these analyses, one should use the \textbf{eligibility time} (the time a subtask \emph{could} run if the executor were free and it had highest priority) as the conceptual \textbf{release time} for the analysis calculation. The term `release time' used elsewhere in this paper refers to the point in time when a subtask is enqueued into the events queue, which occurs at scheduling points.


\section{Evaluations}\label{sec:evaluations}

This evaluation is divided in two parts. In the first part, we implemented the queue detailed in Section~\ref{sec:executor_design} for a RM fixed priority scheduler, as well as one for EDF. 
We first apply the analysis provided by \cite{nasri2019response} to a synthetic taskset. If our executor truly maps to a Limited Preemption Fixed Job-Level Priority scheduler, we would expect their analysis to accurately bound WCRT. This bound is compared to an empirical evaluation is made in Section~\ref{sec:synthetic_evaluation}.

We then present a more empirical measure of performance by assessing the performance of our proposed executor design from Section~\ref{sec:executor_design} compared to the existing events executor, as well as the default ROS~2 executor. We use the Autoware Reference Benchmark, which is intended to compare the performance of different executor designs.

We are limited to a synthetic taskset and the Autoware Reference Benchmark because, as mentioned in Section~\ref{sec:non_lifo_analysis}, no communication middleware available provides LIFO-ordered message queueing, which violates the assumptions found in Lemma~\ref{lemma:lifo_is_priority}.
The synthetic evaluation found in Section~\ref{sec:synthetic_evaluation} is data-agnostic, so by tracing scheduling decisions, the hypothetically LIFO-ordered child subtasks can be reassigned to their correct parents in evaluation post-processing. The Autoware Reference Benchmark in Section~\ref{sec:autoware_evaluation} satisfies the condition in Lemma~\ref{lemma:harmonic_parents}, therefore our analysis holds, according to Theorem~\ref{theorem:executor_works_with_one_element_queue}.

All experiments are ran on a Raspberry Pi Model 4, pinned to a single core with the timer's management thread in the events executor assigned a higher priority, as suggested in \cite{teper2025reconciling}. We report the 99.7th percentile of response times to exclude rare outliers potentially caused by system noise (e.g., OS jitter) rather than by the scheduling algorithm itself. All experiments are conducted for 10 minutes, over multiple hyperperiods.

\subsection{Synthetic Many-to-Many Topic}\label{sec:synthetic_evaluation}

To demonstrate that our analysis holds for an unfolded graph, we replicate the same application shown in Figure~\ref{fig:ros2_unfolding} and vary the WCETs of subtasks to create 3 different utilization levels. Utilization levels are approximate and rounded to even numbers for simplicity. Since this experiment does not use a LIFO-compliant middleware, we take measures to counteract the effect where subsequence jobs from a specific subtask are misordered. Since all data is arbitrary in this synthetic evaluation, jobs of a particular subtask are fungible. That is, if job $x_i$ of subscription subtask $x$ is scheduled to run, it does not matter if it is given the data for job $x_j$ instead, provided the WCET and priority match the intended job that should run.
In post-processing, we identified when each child subtask was released (i.e., when its parent completed) and when it was scheduled. The response time is calculated as the difference between these two timestamps. This correctly measures the scheduling delay introduced by our executor, bypassing the message-data mismatch that would occur from a FIFO message queue.

Each application is centered around a single topic with three publishers and two subscribers. Each publisher has a different period. After unfolding, we can model this application as a forest of 3 trees, each with 1 parent and 2 children. We created three variations of this taskset for different utilization levels of (approximately) 50\%, 70\%, and 90\%. The WCET times for each subtask in each taskset are illustrated in Table~\ref{table:taskset_definitions}.





\begin{table}
  \centering
  \caption{WCET of Subtasks in Synthetic Tasksets}
  \label{table:taskset_definitions}
  \begin{tabular}{c|ccc|cc}
    & \multicolumn{3}{c}{Parent by Period} & \multicolumn{2}{c}{Children} \\
    Utilization & 25ms & 41ms & 51ms & A & B\\
    \hline
    50\% & 1ms & 1ms & 5ms & 2ms & 2ms \\
    70\% & 2ms & 1ms & 12ms & 2ms & 2ms \\
    90\% & 3ms & 1ms & 12ms & 2ms & 4ms \\
  \end{tabular}
\end{table}

The unfolded forest is entered into an analysis tool~\cite{nasri2019response,npscheduabilityanalysis}, which provides a WCRT across the entire tree for each job with both RM and EDF scheduling.
%
%
%
%
\begin{figure}
    \centering
    \includegraphics[width=0.5\linewidth]{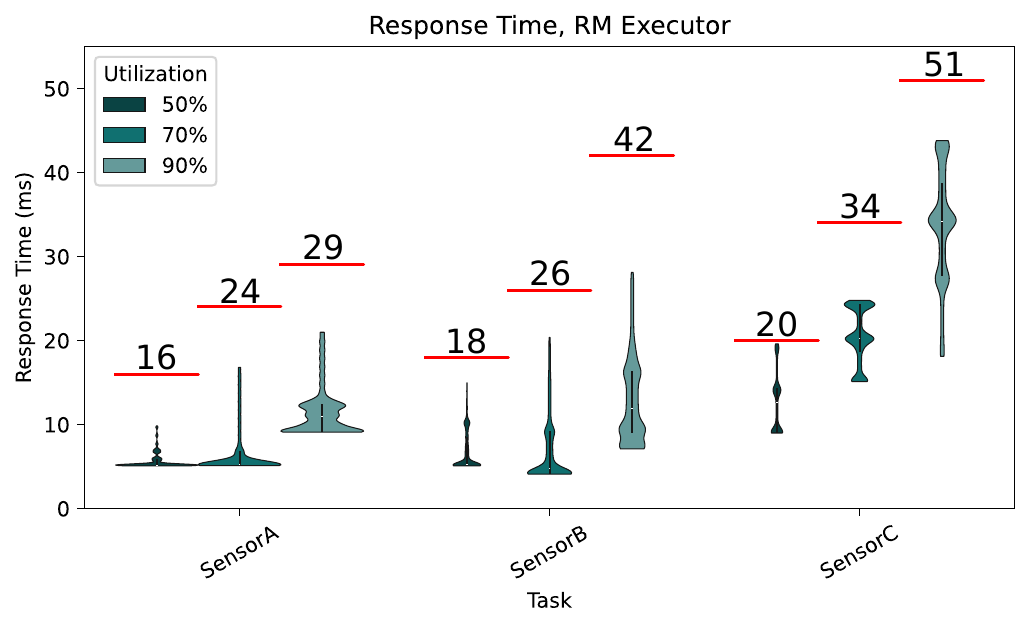}%
    \includegraphics[width=0.5\linewidth]{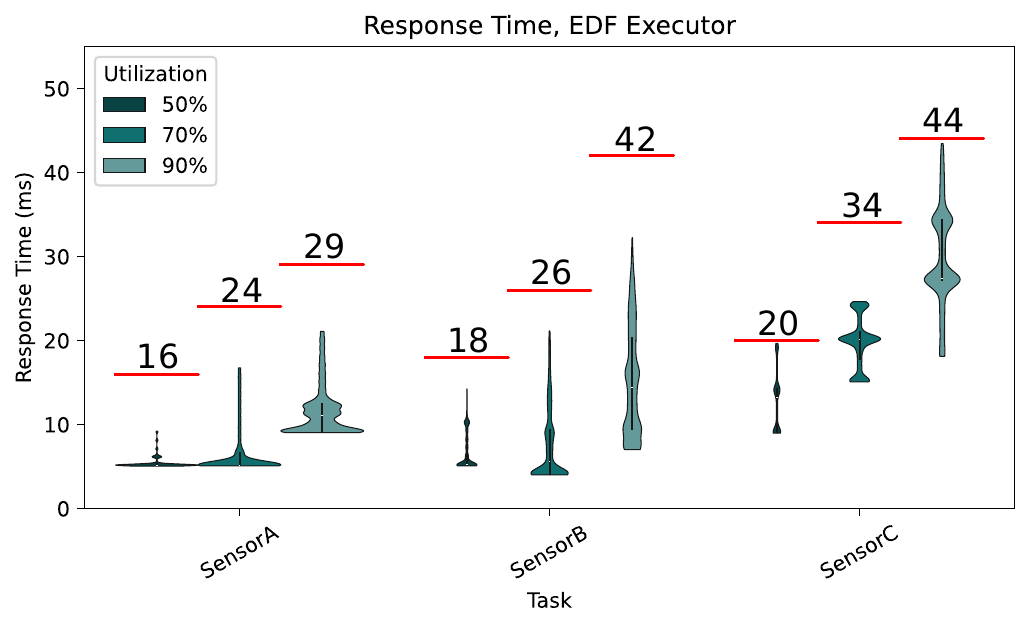}
    \caption{Response times for the synthetic task set compared against the WCRT bounds from analysis~\cite{nasri2019response}. Task A, B, and C correspond to the root tasks with periods 25ms, 41ms, and 51ms, respectively. Red lines indicate the analytical WCRT.}
    \label{fig:many_to_many_latencies}
\end{figure}
Results are shown in Figure~\ref{fig:many_to_many_latencies}. Our measured WCRTs closely match the predicted WCRTs from the analysis. The analysis for the 90\% utilization scenario with RM scheduling predicts deadline misses for Task A and Task B, which highlights that the task set is not schedulable under these conditions. Our executor's behavior aligns with this analysis.

\subsection{Autoware Reference Benchmark}\label{sec:autoware_evaluation}

We also compare our performance to existing executors in the Autoware Reference Benchmark, a simulated version of the Autoware application detailed in~\cite{referencesystem}, which includes interconnected timers and subscriptions.\footnote{\url{https://github.com/ros-realtime/reference-system/blob/main/autoware_reference_system/README.md}}
Specifically, we measure the end-to-end latency of the \textbf{hot path} defined in the benchmark from the front and rear LiDAR sensors to the object collision estimator (c.f. the Autoware reference system), as this latency is a key metric for responsiveness in crash avoidance.
The data at the start of the hot path is generated at a frequency of 10Hz.

Results are shown in Figure~\ref{fig:autoware_evaluation}. The default and events executor perform the worst, but are still consistently within the implicit 100ms deadline of the hotpath. However, with our RM and EDF modifications to the events executor, the worst case is consistently under 25ms. This represents a 63\% improvement over the default executor, and a 61\% improvement over the events executor. In addition to a better WCET, there is less variability as well.

\begin{figure}
    \centering
    \includegraphics[width=.8\linewidth]{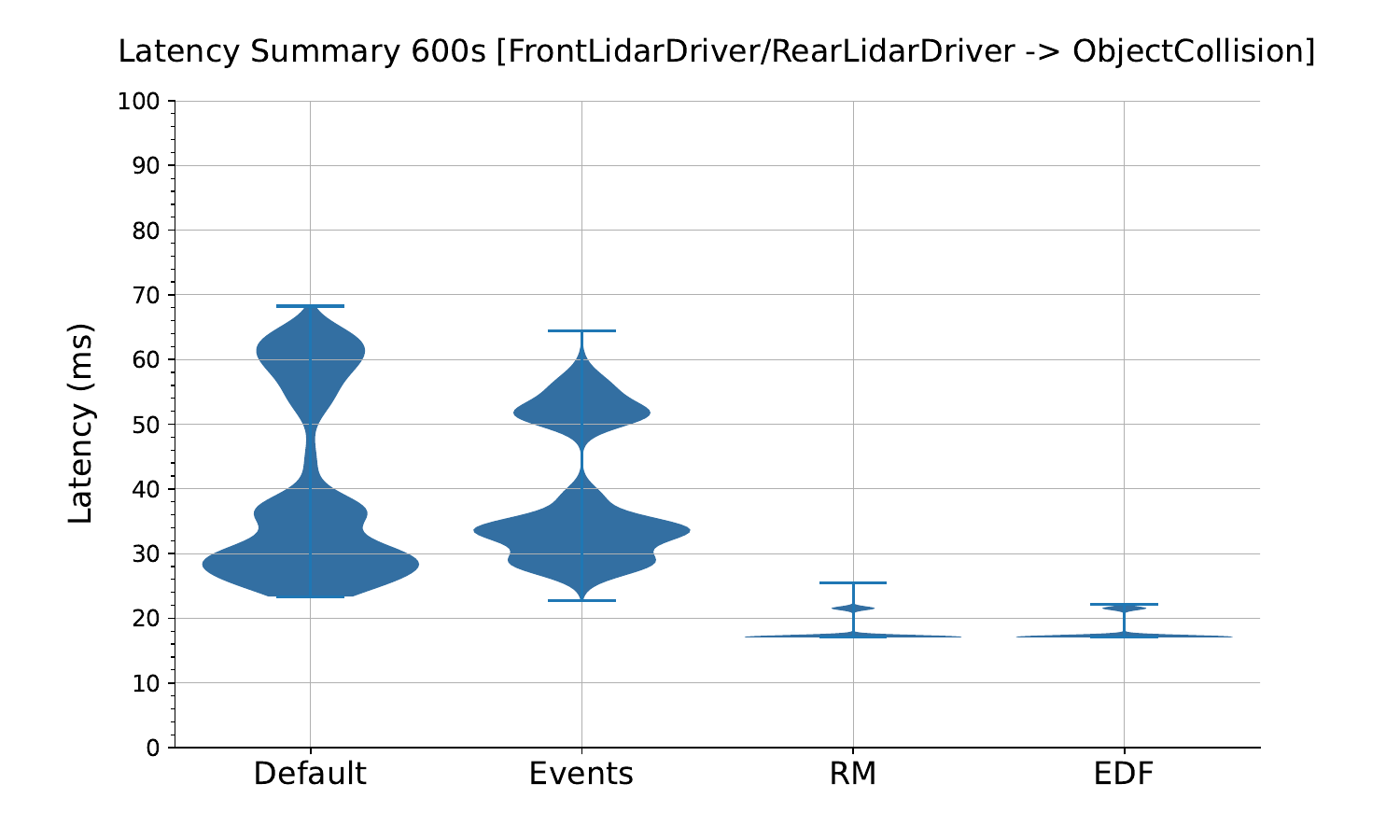}
    \caption{Autoware Reference Benchmark Executor Comparison}
    \label{fig:autoware_evaluation}
\end{figure}

\section{Conclusions and Future Work}\label{sec:conclusions}


Prior literature \cite{choi2021picas,casini2019response,tang2020response,blass2021ros,teper2022end} has catered to the limitations of the default ROS~2 executor. They quantify its behavior and propose multithreaded work-arounds rather than directly addressing the blocking caused by the processing window. In contrast, our work builds on the events executor, which isn't affected by the processing window. Scheduling decisions are made after each executed job, making it compliant with a limited-preemption scheduler. Additionally, the response-time analysis provided by prior literature has focused on chains. Our work not only provides support for general branching graphs, but it does so with minimal retooling required and without the need for a priori application analysis and graph decomposition, as is the case with PiCAS\cite{choi2021picas}.

This paper introduces a novel abstraction with which ROS~2 applications can be analyzed as forests of trees.
It also presents a new design for an events queue that supports arbitrary fixed-job-level-priority schedulers, and naturally honors the precedence relations between subtasks without needing to be aware of them.
Finally, it provides an analysis proving that the events executor using our queue is behaviorally identical to a conventional scheduler operating on the unfolded equivalent of our application. Our analysis and empircal evaluation confirm that the schedules created are identical (subject to arbitrary tie-breaks).
%
%
This further closes the gap between established real-time literature and ROS~2 scheduling analysis. 

The primary limitation to this work is the necessity for a LIFO message queue, which is not to our knowledge provided by any commercially available or open-source DDS middleware. Integrating that feature into an existing open-source middleware framework is a natural direction for future work, to broaden the applicability of the approach presented here beyond the currently targeted set of applications meeting Theorem~\ref{theorem:executor_works_with_one_element_queue}'s conditions (e.g., harmonic periods). Future work could also explore extending this model to mixed-criticality systems or handling other preemption constraints.





\bibliographystyle{abbrv}
\bibliography{references}

\end{document}